\documentclass[a4paper,12pt]{article}

\usepackage{amsmath,amssymb}
\usepackage{amscd}
\usepackage{mathrsfs}
\usepackage{amsthm}
\usepackage{color}
\usepackage{ascmac}

\usepackage{verbatim}

\usepackage{mathtools}
\usepackage{enumerate}

\usepackage[colorlinks=true, linkcolor=blue, citecolor=blue, urlcolor=blue]{hyperref}
\theoremstyle{plain}

\newtheorem{theorem}{\bf Theorem}[section]
\newtheorem{lemma}[theorem]{\bf Lemma}

\newtheorem{corollary}[theorem]{\bf Corollary}

\theoremstyle{definition}

\newtheorem{example}[theorem]{\bf Example}
\newtheorem{remark}[theorem]{\bf Remark}

  \makeatletter
  \newcommand{\subsubsubsection}{\@startsection{paragraph}{4}{\z@}%
    {1.0\Cvs \@plus.5\Cdp \@minus.2\Cdp}%
    {.1\Cvs \@plus.3\Cdp}%
    {\reset@font\sffamily\normalsize}
  }
  \makeatother
\makeatletter

\@addtoreset{equation}{section}
\makeatother

\title{A direct derivation of an effective Hamiltonian \\
	in non-relativistic quantum electrodynamics}

\author{Yasumichi Matsuzawa\thanks{Department of Mathematics, Faculty of Education, Shinshu University, 6-Ro, Nishi-nagano, Nagano 380-8544, Japan, e-mail: myasu@shinshu-u.ac.jp} }
\date{\today}

\begin{document}
\maketitle

\begin{abstract}
	We present a direct derivation of Arai's effective Hamiltonian in non-relativistic quantum electrodynamics without relying on the scaling limit. 
	Our result applies to a broader class of potentials, including the Rollnik class and confining potentials such as the harmonic potential.
\end{abstract}

\section{Introduction}

Non-relativistic quantum electrodynamics provides a framework for describing the interaction between charged particles and the quantized electromagnetic field, and plays a fundamental role in understanding the Lamb shift in hydrogen atoms \cite{PhysRev.72.339}.

Welton \cite{PhysRev.74.1157} explained the Lamb shift based on the physical intuition that vacuum fluctuations of the electromagnetic field induce fluctuations in the position of the electron, which leads to an effective averaging of the potential. 
This picture can be formulated in terms of an effective Hamiltonian in which the electron potential is modified to incorporate the contribution of the electromagnetic field.
Arai derived such an effective Hamiltonian via the scaling limit \cite{MR1075749} and studied its spectral properties to derive the Lamb shift \cite{MR2770092}.

In this paper, we give a direct derivation of the effective Hamiltonian obtained by Arai, without relying on the scaling limit.
We characterize the effective Hamiltonian through expectation values of the total Hamiltonian with respect to dressed electron states.
Our derivation consists of two parts.
First, following Arai \cite{MR1075749}, we construct a dressed electron state for each electron state as a superposition of the ground states of the fiber Hamiltonians weighted by the electron state.
Then, the effective Hamiltonian is characterized as the unique self-adjoint operator whose expectation values with respect to the electron states coincide with those of the total Hamiltonian with respect to the corresponding dressed electron states. 

Compared with \cite{MR1075749}, our result applies to a broader class of potentials, including the Rollnik class and confining potentials such as the harmonic potential. 
In particular, it covers classes of potentials that are not covered by the scaling limit in \cite{MR1075749}, but are treated in \cite{MR2770092}. 
In this sense, it provides a rigorous justification for the study of the effective Hamiltonian in \cite{MR2770092}.

The rest of this paper is organized as follows.
In Section \ref{sect;setup}, we describe our model.
In Section \ref{sect;dressed states}, we construct dressed electron states.
In Section \ref{sect;sa}, we define the total Hamiltonian and the effective Hamiltonian via quadratic form sums.
In Section \ref{sect;derivation}, we derive the effective Hamiltonian from the total Hamiltonian.
In Section \ref{sect;concluding_remarks}, we briefly comment on mass renormalization.

\section{Setup}\label{sect;setup}
We consider the Pauli--Fierz model \cite{PauliFierz1938} describing an electron interacting with the quantized electromagnetic field. 
Following \cite{MR1075749}, we simplify the model as follows:
\begin{itemize}
	\item the dipole approximation is used,
	\item the $\boldsymbol{A}^2$ term is neglected,
	\item mass renormalization of the electron is taken into account.
\end{itemize}
The model is defined within the framework of Fock space. 
For details on Fock space, we refer the reader to \cite{MR4812858}.
Let $d\geq2$ be the spatial dimension.
The Hilbert space $\mathscr{H}_\mathrm{tot}$ of the model is defined by 
\[
\mathscr{H}_\mathrm{tot}:=L^2(\mathbb{R}_{\boldsymbol{x}}^d)\otimes\mathscr{F}_\mathrm{b}(\mathscr{H}_\mathrm{ph}),\qquad \mathscr{H}_\mathrm{ph}:=L^2(\mathbb{R}_{\boldsymbol{k}}^d\times\{1,\dots,d-1\})
\]
where $\mathbb{R}_{\boldsymbol{x}}^d$ (resp. $\mathbb{R}_{\boldsymbol{k}}^d$) denotes $\mathbb{R}^d$ with variable $\boldsymbol{x}$ (resp. variable $\boldsymbol{k}$), and $\mathscr{F}_\mathrm{b}(\mathscr{H}_\mathrm{ph})$ denotes the Boson Fock space over $\mathscr{H}_\mathrm{ph}$. 
The total Hamiltonian $H_\mathrm{tot}$ of the model is defined by
\[
H_\mathrm{tot}:=H_0\dot{+}(V\otimes1),\qquad H_0:=-\frac{\hbar^2}{2m_0}\Delta\otimes1+1\otimes\mathrm{d}\Gamma_\mathrm{b}(\hbar \omega)-\frac{q}{m}\sum_{j=1}^dp_j\otimes A_j(\boldsymbol{0}),
\]
where $\dot{+}$ denotes the quadratic form sum, $\hbar>0$ is the reduced Planck constant, 
$m>0$ is the observed electron mass, 
$m_0>0$ is the bare electron mass, 
$q\in\mathbb{R}\setminus\{0\}$ is the coupling constant.
Each $p_j:=-\mathrm{i}\hbar\partial_{x_j}$ is the momentum operator, where $\partial_{x_j}$ is understood in the distributional sense. 
Moreover, $\Delta:=\sum_{j=1}^d(\partial_{x_j})^2$ is the Laplacian,
$V:\mathbb{R}_{\boldsymbol{x}}^d\to\mathbb{R}$ is a Borel function describing an external potential, 
$\omega:\mathbb{R}_{\boldsymbol{k}}^d\to[0,\infty)$ is a rotationally symmetric Borel function satisfying $\omega(\boldsymbol{k})>0$ for almost every $\boldsymbol{k}\in\mathbb{R}^d$.
Finally, $\mathrm{d}\Gamma_\mathrm{b}(\hbar\omega)$ denotes the second quantization of the multiplication operator by $\hbar\omega$, and each $A_j(\boldsymbol{0})$ is the quantized vector potential at the origin defined below.

To define $A_j(\boldsymbol{0})$, we fix polarization vectors 
\[
\boldsymbol{e}^{(\lambda)}:\mathbb{R}_{\boldsymbol{k}}^d\to\mathbb{R}^d,\qquad \lambda=1,\dots,d-1.
\]
That is, they satisfy
\begin{itemize}
	\item $\boldsymbol{e}^{(\lambda)}$ is Borel measurable for all $\lambda=1,\dots,d-1,$
	\item the set $\{\boldsymbol{k}/|\boldsymbol{k}|,\boldsymbol{e}^{(1)}(\boldsymbol{k}),\dots,\boldsymbol{e}^{(d-1)}(\boldsymbol{k})\}$ is an orthonormal basis of $\mathbb{R}^d$ for almost every $\boldsymbol{k}\in\mathbb{R}^d$.
\end{itemize}
Let $\hat{\rho}:\mathbb{R}_{\boldsymbol{k}}^d\to\mathbb{R}$ be a Borel function depending only on $|\boldsymbol{k}|$ such that $\hat{\rho}/\omega^{1/2},\hat{\rho}/\omega^{3/2}\in L^2(\mathbb{R}_{\boldsymbol{k}}^d)\setminus\{0\}$.
For each $j=1,\dots,d$, we define $g_j\in\mathscr{H}_\mathrm{ph}$ by
\[
g_j(\boldsymbol{k},\lambda):=\sqrt{\frac{\hbar}{\epsilon_0\omega(\boldsymbol{k})}}\hat{\rho}(\boldsymbol{k})e_j^{(\lambda)}(\boldsymbol{k}),\qquad \boldsymbol{k}\in\mathbb{R}^d,\ \lambda=1,\dots,d-1,
\]
where $\epsilon_0>0$ denotes the vacuum permittivity and $e_j^{(\lambda)}(\boldsymbol{k})$ denotes the $j$-th component of $\boldsymbol{e}^{(\lambda)}(\boldsymbol{k})$.
We then define $A_j(\boldsymbol{0}):=\Phi_\mathrm{S}(g_j)$, where $\Phi_\mathrm{S}(f)$ denotes the Segal field operator with test vector $f\in\mathscr{H}_\mathrm{ph}$.

The following standard estimate (see, e.g., \cite[Proposition 5.16]{MR4812858}) will be used below:
\begin{equation}\label{standard estimate}
	\|\Phi_\mathrm{S}(f)\Psi\|\leq\sqrt{2}\|(\hbar \omega)^{-1/2}f\|\|\mathrm{d}\Gamma_\mathrm{b}(\hbar \omega)^{1/2}\Psi\|+\frac{1}{\sqrt{2}}\|f\|\|\Psi\|
\end{equation}
for all $f\in\mathrm{dom}(\omega^{-1/2})$ and $\Psi\in\mathrm{dom}(\mathrm{d}\Gamma_\mathrm{b}(\hbar \omega)^{1/2})$.

\begin{remark}
	Physically, $d=3$ and $\omega(\boldsymbol{k})=c|\boldsymbol{k}|$, where $c>0$ is the speed of light.
	An important example of $\hat{\rho}$ used in the derivation of the Lamb shift \cite{MR2770092,PhysRev.74.1157} is 
	\[
	\hat{\rho}(\boldsymbol{k})=\begin{cases*}
		(2\pi)^{-3/2} & if $\kappa\leq|\boldsymbol{k}|\leq mc/\hbar$, \\
		0                 & otherwise,
	\end{cases*}
	\]
	where $\kappa$ is a constant satisfying $0<\kappa<mc/\hbar$.
\end{remark}

The relation between $m_0$ and $m$ is given by
\begin{equation}\label{relation btw m_0 and m}
\frac{1}{m_0}=\frac{1}{m}+\frac{d-1}{d}\frac{q^2}{\epsilon_0m^2}\int_{\mathbb{R}^d}\frac{|\hat{\rho}(\boldsymbol{k})|^2}{\omega(\boldsymbol{k})^2}\,\mathrm{d}\boldsymbol{k}.
\end{equation}
We further assume the following conditions on the potential $V$:
\begin{itemize}
	\item[\textup{(V.1)}] $V$ satisfies either of the following conditions:
	\begin{itemize}
		\item[\textup{(V.1-1)}] $|V|$ is infinitesimally form-bounded with respect to $-\Delta$.
		\item[\textup{(V.1-2)}] $V$ is bounded from below.
	\end{itemize} 
	\item[\textup{(V.2)}] $\int_{\mathbb{R}^d}|V(\boldsymbol{x})|\mathrm{e}^{-t|\boldsymbol{x}|^2}\,\mathrm{d}\boldsymbol{x}<\infty$ for all $t>0$. 
\end{itemize}

\begin{example}
	\begin{itemize}
		\item[(1)] Let $d=3$, and let $\mathcal{R}$ be the Rollnik class.
		That is, $\mathcal{R}$ is the set of Borel functions $V:\mathbb{R}_{\boldsymbol{x}}^3\to\mathbb{R}$ satisfying
		\[
		\int_{\mathbb{R}^3}\int_{\mathbb{R}^3}\frac{|V(\boldsymbol{x})||V(\boldsymbol{y})|}{|\boldsymbol{x}-\boldsymbol{y}|^2}\,\mathrm{d}\boldsymbol{x}\,\mathrm{d}\boldsymbol{y}<\infty.
		\]
		Then, every $V\in\mathcal{R}$ satisfies (V.1-1) and (V.2).
		Thus, so does every $V\in\mathcal{R}+L^\infty(\mathbb{R}_{\boldsymbol{x}}^3)$.
		\item[(2)] A harmonic potential $V(\boldsymbol{x}):=K|\boldsymbol{x}|^2/2$ on $\mathbb{R}^d$, where $K>0$ is a constant, satisfies (V.1-2) and (V.2).
	\end{itemize}
\end{example}

\begin{proof}
	We have only to prove that any $V\in\mathcal{R}$ satisfies (V.2), since every potential in $\mathcal{R}$ satisfies (V.1-1) \cite[Theorem 1.21]{MR455975}.
	Let $t>0$.
	Since every potential in $\mathcal{R}$ belongs to $L^1_\mathrm{loc}(\mathbb{R}_{\boldsymbol{x}}^3)$ \cite[Theorem 1.7]{MR455975}, it suffices to show that
	\begin{equation}\label{eq;Rolnik_implies_(V.2)}
	\int_{|\boldsymbol{x}|\geq R}|V(\boldsymbol{x})|\mathrm{e}^{-t|\boldsymbol{x}|^2}\,\mathrm{d}\boldsymbol{x}<\infty
	\end{equation} for some $R>0$.
	Since $V\in\mathcal{R}$, there exists $\boldsymbol{y}\in\mathbb{R}^3$ such that $\int_{\mathbb{R}^3}|V(\boldsymbol{x})|/|\boldsymbol{x}-\boldsymbol{y}|^2\,\mathrm{d}\boldsymbol{x}<\infty$.
	Take a sufficiently large $R>0$ so that for all $|\boldsymbol{x}|\geq R$, $\mathrm{e}^{-t|\boldsymbol{x}|^2}\leq1/|\boldsymbol{x}-\boldsymbol{y}|^2$ holds.
	Combining the above estimates, we obtain \eqref{eq;Rolnik_implies_(V.2)}.
\end{proof}

\section{Construction of dressed electron states}\label{sect;dressed states}

In this section, we recall the necessary parts of the arguments in \cite[Section III]{MR1075749}.
We consider the Hamiltonian $H_0$ corresponding to the case $V=0$.
Let $\mathcal{F}_d:L^2(\mathbb{R}_{\boldsymbol{x}}^d)\to L^2(\mathbb{R}_{\boldsymbol{P}}^d)$ be the Fourier transform defined by
\[
(\mathcal{F}_d\psi)(\boldsymbol{P}):=\frac{1}{(2\pi\hbar)^{d/2}}\int_{\mathbb{R}^d}\psi(\boldsymbol{x})\mathrm{e}^{-\mathrm{i}\boldsymbol{x}\cdot\boldsymbol{P}/\hbar}\,\mathrm{d}\boldsymbol{x},\qquad \psi\in L^2(\mathbb{R}_{\boldsymbol{x}}^d),\ \boldsymbol{P}\in\mathbb{R}^d
\]
in the $L^2$-sense.
Under the natural isomorphism
\[
L^2(\mathbb{R}_{\boldsymbol{P}}^d)\otimes\mathscr{F}_\mathrm{b}(\mathscr{H}_\mathrm{ph})
\cong L^2(\mathbb{R}_{\boldsymbol{P}}^d;\mathscr{F}_\mathrm{b}(\mathscr{H}_\mathrm{ph}))
=\int_{\mathbb{R}^d}^\oplus\mathscr{F}_\mathrm{b}(\mathscr{H}_\mathrm{ph})\,\mathrm{d}\boldsymbol{P},
\]
we obtain
\begin{equation}\label{eq;fiber decomposition}
(\mathcal{F}_d\otimes1)H_0(\mathcal{F}_d\otimes1)^*=\int_{\mathbb{R}^d}^\oplus H_0(\boldsymbol{P})\,\mathrm{d}\boldsymbol{P},
\end{equation}
where the fiber Hamiltonian $H_0(\boldsymbol{P})$ is given by
\[
H_0(\boldsymbol{P}):=\mathrm{d}\Gamma_\mathrm{b}(\hbar\omega)-\frac{q}{m}\Phi_\mathrm{S}(\boldsymbol{P}\cdot\boldsymbol{g})+\frac{|\boldsymbol{P}|^2}{2m_0},\qquad
\boldsymbol{P}\cdot\boldsymbol{g}:=\sum_{j=1}^dP_jg_j.
\]
$H_0(\boldsymbol{P})$ is a van Hove Hamiltonian; see, e.g., \cite[Chapter 13]{MR4812858} or \cite{MR2015428} for its basic properties.
The Kato--Rellich theorem, combined with the standard estimate \eqref{standard estimate}, implies that
$H_0(\boldsymbol{P})$ is self-adjoint.
Moreover, it holds that
\begin{equation*}
U(\boldsymbol{P})H_0(\boldsymbol{P})U(\boldsymbol{P})^*
=\mathrm{d}\Gamma_\mathrm{b}(\hbar\omega)+E_0(\boldsymbol{P}),\qquad U(\boldsymbol{P}):=\mathrm{e}^{\mathrm{i}q/(m\hbar)\Phi_\mathrm{S}(\mathrm{i}\boldsymbol{P}\cdot\boldsymbol{g}/\omega)},
\end{equation*}
where the ground state energy $E_0(\boldsymbol{P})$ is given by
\[
E_0(\boldsymbol{P}):=-\frac{q^2}{2m^2\hbar}\left\|\frac{\boldsymbol{P}\cdot\boldsymbol{g}}{\sqrt{\omega}}\right\|^2+\frac{|\boldsymbol{P}|^2}{2m_0}.
\]
Thus, $H_0(\boldsymbol{P})$ admits a unique ground state $\Phi_0(\boldsymbol{P})$ given by $\Phi_0(\boldsymbol{P}):=U(\boldsymbol{P})^*\Omega$,
where $\Omega\in\mathscr{F}_\mathrm{b}(\mathscr{H}_\mathrm{ph})$ is the Fock vacuum defined by $\Omega:=\{1,0,0,\dots\}$.
It follows from \eqref{relation btw m_0 and m} that $E_0(\boldsymbol{P})=|\boldsymbol{P}|^2/(2m)$.

As a corollary, we obtain
\begin{equation}\label{eq;diag of H_0}
UH_0U^*=-\frac{\hbar^2}{2m}\Delta\otimes1+1\otimes\mathrm{d}\Gamma_\mathrm{b}(\hbar\omega),
\end{equation}
where the unitary operator $U$ on $\mathscr{H}_\mathrm{tot}$ is defined by
\[
U:=(\mathcal{F}_d\otimes1)^*\left[\int_{\mathbb{R}^d}^\oplus U(\boldsymbol{P})\,\mathrm{d}\boldsymbol{P}\right](\mathcal{F}_d\otimes1).
\]
Thus, $H_0$ is self-adjoint and non-negative.

\begin{remark}
	The operator $U$ coincides with the operator denoted by $\mathrm{e}^{\mathrm{i}T}$ in \cite[Section III]{MR1075749}.
\end{remark}

For each $\psi\in L^2(\mathbb{R}_{\boldsymbol{x}}^d)$, we define a dressed electron state $\psi_\mathrm{dr}\in\mathscr{H}_\mathrm{tot}$ by $\psi_\mathrm{dr}:=U^*(\psi\otimes\Omega)$.
Under the natural isomorphism
\[
L^2(\mathbb{R}_{\boldsymbol{x}}^d)\otimes\mathscr{F}_\mathrm{b}(\mathscr{H}_\mathrm{ph})
\cong L^2(\mathbb{R}_{\boldsymbol{x}}^d;\mathscr{F}_\mathrm{b}(\mathscr{H}_\mathrm{ph})),
\]
it holds that
\begin{equation}\label{superposition rep}
\psi_\mathrm{dr}(\boldsymbol{x})=\frac{1}{(2\pi\hbar)^{d/2}}\int_{\mathbb{R}^d}(\mathcal{F}_d\psi)(\boldsymbol{P})\Phi_0(\boldsymbol{P})\mathrm{e}^{\mathrm{i}\boldsymbol{x}\cdot \boldsymbol{P}/\hbar}\,\mathrm{d}\boldsymbol{P},\qquad \mathrm{a.e.}\,\boldsymbol{x}\in\mathbb{R}^d,
\end{equation}
provided that $\mathcal{F}_d\psi\in L^1(\mathbb{R}_{\boldsymbol{P}}^d)$. 
In this sense, $\psi_\mathrm{dr}$ is a superposition of the ground states $\Phi_0(\boldsymbol{P})$ of the fiber Hamiltonians $H_0(\boldsymbol{P})$ weighted by the electron state $\psi$.

\section{Self-adjoint realization of the Hamiltonians}\label{sect;sa}

We use the theory of quadratic forms to define self-adjoint operators.
For details on quadratic forms, we refer the reader to \cite{MR493420,MR2953553}.

\subsection{The total Hamiltonian}

\begin{theorem}\label{thm;tot_ham}
	The quadratic form sum $H_\mathrm{tot}=H_0\dot{+}(V\otimes1)$ exists.
	Consequently, $H_\mathrm{tot}$ is self-adjoint and bounded from below.
\end{theorem}

\begin{proof}
We first consider the case where $V$ satisfies (V.1-1).
We show that $V\otimes1$ is infinitesimally form-bounded with respect to $H_0$.
Since $|V|$ is infinitesimally form-bounded with respect to $-\Delta$, for any $\alpha>0$, there exists $\beta\geq0$ such that
\[
\|(|V|\otimes1)^{1/2}\Psi\|^2\leq \alpha\left\|\left[-\hbar^2(2m)^{-1}\Delta\otimes1\right]^{1/2}\Psi\right\|^2+\beta\|\Psi\|^2,\qquad \forall\Psi\in\mathrm{dom}(|\Delta|^{1/2}\otimes1).
\]
The equalities $U(-\Delta\otimes1)^{1/2}U^*=(-\Delta\otimes1)^{1/2}$ and \eqref{eq;diag of H_0} yield that 
\begin{align*}
\left\|\left[-\hbar^2(2m)^{-1}\Delta\otimes1\right]^{1/2}\Psi\right\|
&=\left\|\left[-\hbar^2(2m)^{-1}\Delta\otimes1\right]^{1/2}U\Psi\right\|\\
&\leq\left\|\left[-\hbar^2(2m)^{-1}\Delta\otimes1+1\otimes\mathrm{d}\Gamma_\mathrm{b}(\hbar\omega)\right]^{1/2}U\Psi\right\|
=\|H_0^{1/2}\Psi\|
\end{align*}
for all $\Psi\in\mathrm{dom}(H_0^{1/2})$.
Hence
\[
\|(|V|\otimes1)^{1/2}\Psi\|^2\leq\alpha\|H_0^{1/2}\Psi\|^2+\beta\|\Psi\|^2,\qquad \forall\Psi\in\mathrm{dom}(H_0^{1/2}).
\]
This implies that $|V|\otimes1$ (and thus $V\otimes1$) is infinitesimally form-bounded with respect to $H_0$.

We next consider the case where $V$ satisfies (V.1-2).
Since $H_0$ and $V\otimes1$ are bounded from below, it suffices to show that $\mathrm{dom}(H_0^{1/2})\cap\mathrm{dom}(|V|^{1/2}\otimes1)$ is dense.
By the definition of $H_0$, we have
\[
C_0^\infty(\mathbb{R}_{\boldsymbol{x}}^d)\hat{\otimes}\mathrm{dom}(\mathrm{d}\Gamma_\mathrm{b}(\hbar\omega))
\subset \mathrm{dom}(H_0)\subset\mathrm{dom}(H_0^{1/2}),
\]
where $C_0^\infty(\mathbb{R}_{\boldsymbol{x}}^d)$ denotes the space of $C^\infty$-functions with compact support, and $\hat{\otimes}$ denotes the algebraic tensor product.
On the other hand, by the condition (V.2), $V$ is in $L_\mathrm{loc}^1(\mathbb{R}_{\boldsymbol{x}}^d)$, and thus $\mathrm{dom}(|V|^{1/2})$ contains $C_0^\infty(\mathbb{R}_{\boldsymbol{x}}^d)$. 
Hence, $\mathrm{dom}(|V|^{1/2}\otimes1)$ also contains $C_0^\infty(\mathbb{R}_{\boldsymbol{x}}^d)\hat{\otimes}\mathrm{dom}(\mathrm{d}\Gamma_\mathrm{b}(\hbar\omega))$.
\end{proof}

\subsection{The effective Hamiltonian}

We define the effective potential $V_\mathrm{eff}:\mathbb{R}_{\boldsymbol{x}}^d\to\mathbb{R}$ by
\begin{equation}\label{eq;def_of_eff_potential}
V_\mathrm{eff}(\boldsymbol{x}):=\frac{1}{(4\pi a)^{d/2}}\int_{\mathbb{R}^d}V(\boldsymbol{y})\mathrm{e}^{-|\boldsymbol{x}-\boldsymbol{y}|^2/(4a)}\,\mathrm{d}\boldsymbol{y},\qquad \boldsymbol{x}\in\mathbb{R}^d
\end{equation}
with
\begin{equation}\label{eq;def_of_a}
a:=\frac{d-1}{4d}\frac{\hbar q^2}{\epsilon_0m^2}\int_{\mathbb{R}^d}\frac{|\hat{\rho}(\boldsymbol{k})|^2}{\omega(\boldsymbol{k})^3}\,\mathrm{d}\boldsymbol{k}.
\end{equation}
Note that, by the condition (V.2), the right-hand side of \eqref{eq;def_of_eff_potential} absolutely converges and $V_\mathrm{eff}$ is continuous on $\mathbb{R}^d$. 

\begin{theorem}\label{thm;eff hamiltonian}
The quadratic form sum
\begin{equation}\label{eq;def_of_Heff}
H_\mathrm{eff}:=-\frac{\hbar^2}{2m}\Delta\dot{+}V_\mathrm{eff}
\end{equation}
exists.
Consequently, $H_\mathrm{eff}$ is self-adjoint and bounded from below.
\end{theorem}

The operator $H_\mathrm{eff}$ is referred to as the effective Hamiltonian \cite{MR2770092}.
We need two lemmas to prove Theorem \ref{thm;eff hamiltonian}.

\begin{lemma}\label{lem;eff1}
	Let $\psi\in L^2(\mathbb{R}_{\boldsymbol{x}}^d)$.
	Then, $\psi_\mathrm{dr}\in\mathrm{dom}(H_0^{1/2})$ if and only if $\psi\in\mathrm{dom}(|\Delta|^{1/2})$.
	In this case, we have
	\[
	\|H_0^{1/2}\psi_\mathrm{dr}\|^2=\frac{\hbar^2}{2m}\|(-\Delta)^{1/2}\psi\|^2.
	\]
\end{lemma}
\begin{proof}
	By the fiber decomposition \eqref{eq;fiber decomposition} and the definition of $U$ and $\psi_\mathrm{dr}$, we have
	\begin{align*}
		\|H_0^{1/2}\psi_\mathrm{dr}\|^2
		&=\int_{\mathbb{R}^d}\left\|H_0(\boldsymbol{P})^{1/2}(\mathcal{F}_d\psi)(\boldsymbol{P})\Phi_0(\boldsymbol{P})\right\|_{\mathscr{F}_\mathrm{b}(\mathscr{H}_\mathrm{ph})}^2\,\mathrm{d}\boldsymbol{P}\\
		&=\int_{\mathbb{R}^d}E_0(\boldsymbol{P})|(\mathcal{F}_d\psi)(\boldsymbol{P})|^2\,\mathrm{d}\boldsymbol{P}
		=\frac{\hbar^2}{2m}\|(-\Delta)^{1/2}\psi\|^2.
	\end{align*}
	This completes the proof.
\end{proof}

\begin{lemma}\label{lem;eff2}
	Let $\psi\in L^2(\mathbb{R}_{\boldsymbol{x}}^d)$.
	Then, $\psi_\mathrm{dr}\in\mathrm{dom}(|V|^{1/2}\otimes1)$ if and only if $\psi\in\mathrm{dom}\left((|V|_\mathrm{eff})^{1/2}\right)$.
	In this case, we have $\psi\in\mathrm{dom}(|V_\mathrm{eff}|^{1/2})$ and
	\begin{equation}\label{eq;potential}
	\int_{\mathbb{R}^d}V(\boldsymbol{x})\|\psi_\mathrm{dr}(\boldsymbol{x})\|_{\mathscr{F}_\mathrm{b}(\mathscr{H}_\mathrm{ph})}^2\,\mathrm{d}\boldsymbol{x}
	=\int_{\mathbb{R}^d}V_\mathrm{eff}(\boldsymbol{x})|\psi(\boldsymbol{x})|^2\,\mathrm{d}\boldsymbol{x}.
	\end{equation}
\end{lemma}

\begin{proof}
	We first consider the case where $V$ is bounded and non-negative.
	In this case, $V_\mathrm{eff}$ is also bounded and non-negative. 
	For a Schwartz function $\psi$, it follows from \eqref{superposition rep} that
	\begin{align*}
		\int_{\mathbb{R}^d}V(\boldsymbol{x})\|\psi_\mathrm{dr}(\boldsymbol{x})\|_{\mathscr{F}_\mathrm{b}(\mathscr{H}_\mathrm{ph})}^2\,\mathrm{d}\boldsymbol{x}
		&=\frac{1}{(2\pi\hbar)^{d}}\int_{\mathbb{R}^d}\int_{\mathbb{R}^d}\int_{\mathbb{R}^d}
		V(\boldsymbol{x})\overline{(\mathcal{F}_d\psi)(\boldsymbol{P})}(\mathcal{F}_d\psi)(\boldsymbol{P}^\prime)\\
		&\qquad\times\left\langle\Phi_0(\boldsymbol{P}),\Phi_0(\boldsymbol{P}^\prime)\right\rangle_{\mathscr{F}_\mathrm{b}(\mathscr{H}_\mathrm{ph})}\mathrm{e}^{-\mathrm{i}\boldsymbol{x}\cdot (\boldsymbol{P}-\boldsymbol{P}^\prime)/\hbar}
		\,\mathrm{d}\boldsymbol{P}\,\mathrm{d}\boldsymbol{P^\prime}\,\mathrm{d}\boldsymbol{x}.
	\end{align*}
	We set
	\[
	G(\boldsymbol{x}):=\frac{1}{(4\pi a)^{d/2}}\mathrm{e}^{-|\boldsymbol{x}|^2/(4a)},\qquad \boldsymbol{x}\in\mathbb{R}^d,
	\]
	which satisfies
	\[
	\langle\psi,G(\cdot-\boldsymbol{x})\psi\rangle
	=\langle\mathcal{F}_d\psi,\mathcal{F}_d[G(\cdot-\boldsymbol{x})\psi]\rangle
	=\frac{1}{(2\pi\hbar)^{d/2}}\langle\mathcal{F}_d\psi,[\mathcal{F}_dG(\cdot-\boldsymbol{x})]\ast[\mathcal{F}_d\psi]\rangle
	\]
	and 
	\[
	[\mathcal{F}_dG(\cdot-\boldsymbol{x})](\boldsymbol{P})
	=\frac{1}{(2\pi\hbar)^{d/2}}\mathrm{e}^{-a|\boldsymbol{P}|^2/\hbar^2}\mathrm{e}^{-\mathrm{i}\boldsymbol{x}\cdot\boldsymbol{P}/\hbar},\qquad\boldsymbol{P}\in\mathbb{R}^d.
	\]
	Then, we obtain
	\begin{align*}
		\int_{\mathbb{R}^d}V_\mathrm{eff}(\boldsymbol{y})|\psi(\boldsymbol{y})|^2\,\mathrm{d}\boldsymbol{y}
		&=\int_{\mathbb{R}^d}\int_{\mathbb{R}^d}|\psi(\boldsymbol{y})|^2V(\boldsymbol{x})G(\boldsymbol{y}-\boldsymbol{x})\,\mathrm{d}\boldsymbol{x}\,\mathrm{d}\boldsymbol{y}\\
		&=\int_{\mathbb{R}^d}V(\boldsymbol{x})\langle\psi,G(\cdot-\boldsymbol{x})\psi\rangle\,\mathrm{d}\boldsymbol{x}\\
		&=\frac{1}{(2\pi\hbar)^{d}}\int_{\mathbb{R}^d}\int_{\mathbb{R}^d}\int_{\mathbb{R}^d}
		V(\boldsymbol{x})\overline{(\mathcal{F}_d\psi)(\boldsymbol{P})}(\mathcal{F}_d\psi)(\boldsymbol{P}^\prime)\\
		&\qquad\qquad\qquad\qquad\times\mathrm{e}^{-a|\boldsymbol{P}-\boldsymbol{P}^\prime|^2/\hbar^2}\mathrm{e}^{-\mathrm{i}\boldsymbol{x}\cdot (\boldsymbol{P}-\boldsymbol{P}^\prime)/\hbar}
		\,\mathrm{d}\boldsymbol{P^\prime}\,\mathrm{d}\boldsymbol{P}\,\mathrm{d}\boldsymbol{x}.
	\end{align*}
	This, combined with the fact that
	\[
	\left\langle\Phi_0(\boldsymbol{P}),\Phi_0(\boldsymbol{P}^\prime)\right\rangle
	=\mathrm{e}^{-q^2\|(\boldsymbol{P}-\boldsymbol{P}^\prime)\cdot\boldsymbol{g}/\omega\|^2/(4m^2\hbar^2)}
	=\mathrm{e}^{-a|\boldsymbol{P}-\boldsymbol{P}^\prime|^2/\hbar^2},
	\]
	implies \eqref{eq;potential}.
	By a limiting argument, \eqref{eq;potential} holds for all $\psi\in L^2(\mathbb{R}_{\boldsymbol{x}}^d)$.
	
	We next consider the case where $V$ is non-negative.
	In this case, $V_\mathrm{eff}$ is also non-negative.
	For each $n\in\mathbb{N}$, we define $V_n\in L^\infty(\mathbb{R}_{\boldsymbol{x}}^d)$ by
	\[
	V_n(\boldsymbol{x}):=\begin{cases*}
		V(\boldsymbol{x}) & if $V(\boldsymbol{x})\leq n$, \\
		0                 & if $V(\boldsymbol{x})> n$.
	\end{cases*}
	\]
	Since $V_n$ satisfies the conditions (V.1-2) and (V.2),
	it follows from the previous paragraph that
	\[
	\int_{\mathbb{R}^d}V_n(\boldsymbol{x})\|\psi_\mathrm{dr}(\boldsymbol{x})\|^2\,\mathrm{d}\boldsymbol{x}
	=\int_{\mathbb{R}^d}(V_n)_\mathrm{eff}(\boldsymbol{x})|\psi(\boldsymbol{x})|^2\,\mathrm{d}\boldsymbol{x},\qquad\forall\psi\in L^2(\mathbb{R}_{\boldsymbol{x}}^d),\ n\in\mathbb{N}.
	\]
	Since $V_n$ and $(V_n)_\mathrm{eff}$ are non-decreasing on $n\in\mathbb{N}$, the monotone convergence theorem yields that \eqref{eq;potential} holds for all $\psi\in L^2(\mathbb{R}_{\boldsymbol{x}}^d)$, where both sides may be infinite.
	Thus, the desired results follow.
	
	Finally, we consider the general case.
	Since $|V|$ satisfies the conditions \mbox{(V.1-2)} and (V.2), it follows from the previous paragraph that for any $\psi\in L^2(\mathbb{R}_{\boldsymbol{x}}^d)$, $\psi_\mathrm{dr}\in\mathrm{dom}(|V|^{1/2}\otimes1)$ is equivalent to $\psi\in\mathrm{dom}\left((|V|_\mathrm{eff})^{1/2}\right)$.
	To prove \eqref{eq;potential} for $\psi\in\mathrm{dom}\left((|V|_\mathrm{eff})^{1/2}\right)$, we define $V_\pm:=(|V|\pm V)/2$.
	Then, $V_+$ and $V_-$ satisfy the conditions (V.1-2) and (V.2), and are non-negative.
	From $|V|=V_++V_-$ and the previous paragraph, we obtain $\psi_\mathrm{dr}\in\mathrm{dom}((V_\pm)^{1/2}\otimes1)$ and
	\[
	\int_{\mathbb{R}^d}V_\pm(\boldsymbol{x})\|\psi_\mathrm{dr}(\boldsymbol{x})\|_{\mathscr{F}_\mathrm{b}(\mathscr{H}_\mathrm{ph})}^2\,\mathrm{d}\boldsymbol{x}
	=\int_{\mathbb{R}^d}(V_\pm)_\mathrm{eff}(\boldsymbol{x})|\psi(\boldsymbol{x})|^2\,\mathrm{d}\boldsymbol{x}.
	\]
	Subtracting the two equations, we obtain $\psi\in\mathrm{dom}(|V_\mathrm{eff}|^{1/2})$ and \eqref{eq;potential}.
	This completes the proof.
\end{proof}

\begin{proof}[Proof of Theorem \ref{thm;eff hamiltonian}]
	We first consider the case where $V$ satisfies (V.1-1).
	By the proof of Theorem \ref{thm;tot_ham}, $V\otimes1$ is infinitesimally form-bounded with respect to $H_0$.
	This, together with Lemma \ref{lem;eff1} and Lemma \ref{lem;eff2}, yields that $V_\mathrm{eff}$ is infinitesimally form-bounded with respect to $-\hbar^2(2m)^{-1}\Delta$.
	Thus, the quadratic form sum $H_\mathrm{eff}$ is defined.
	
	We next consider the case where $V$ satisfies (V.1-2).
	Since $V$ is bounded from below, so is $V_\mathrm{eff}$.
	Moreover, since $V_\mathrm{eff}$ is continuous, $V_\mathrm{eff}$ is in $L_\mathrm{loc}^1(\mathbb{R}_{\boldsymbol{x}}^d)$.
	Thus, the quadratic form sum $H_\mathrm{eff}$ is defined.
\end{proof}

\begin{example}
	We consider the harmonic potential $V(\boldsymbol{x}):=K|\boldsymbol{x}|^2/2$ on $\mathbb{R}^d$, where $K>0$ is a constant.
	By a straightforward computation, we have
	\[
	V_\mathrm{eff}(\boldsymbol{x})=\frac{K}{2}|\boldsymbol{x}|^2+Kda,\qquad\forall\boldsymbol{x}\in\mathbb{R}^d.
	\]
\end{example}

\section{Derivation of the effective Hamiltonian}\label{sect;derivation}

Let $s_\mathrm{tot}$ be the quadratic form associated with $H_\mathrm{tot}$.
By Lemma \ref{lem;eff1} and Lemma \ref{lem;eff2}, we can define a quadratic form $s_\mathrm{eff}$ by
\[
s_\mathrm{eff}(\psi):=s_\mathrm{tot}(\psi_\mathrm{dr}),\qquad \psi\in\mathrm{dom}(|\Delta|^{1/2})\cap\mathrm{dom}\left((|V|_\mathrm{eff})^{1/2}\right).
\]
Then, the effective Hamiltonian $H_\mathrm{eff}$ is characterized as follows:

\begin{theorem}\label{thm;derivation}
	 The quadratic form $s_\mathrm{eff}$ is closed and bounded from below, and the self-adjoint operator associated with $s_\mathrm{eff}$ is $H_\mathrm{eff}$.
\end{theorem}

We need the following lemma.

\begin{lemma}\label{lem;form_dom}
	It holds that
	\begin{equation}\label{eq;form_dom}
		\mathrm{dom}(|\Delta|^{1/2})\cap\mathrm{dom}\left((|V|_\mathrm{eff})^{1/2}\right)
		=\mathrm{dom}(|\Delta|^{1/2})\cap\mathrm{dom}(|V_\mathrm{eff}|^{1/2}).
	\end{equation}
\end{lemma}

\begin{proof}
	We first consider the case where $V$ satisfies the condition (V.1-1).
	Since $|V|$ also satisfies (V.1-1) and (V.2), it follows from the proof of Theorem \ref{thm;eff hamiltonian} that
	$V_\mathrm{eff}$ and $|V|_\mathrm{eff}$ are infinitesimally form-bounded with respect to $-\hbar^2(2m)^{-1}\Delta$.
	Thus, both sides of \eqref{eq;form_dom} coincide with $\mathrm{dom}(|\Delta|^{1/2})$.
	
	We next consider the case where $V$ satisfies the condition (V.1-2).
    Define $V_\pm:=(|V|\pm V)/2$.
    Then, $V_-$ and $(V_-)_\mathrm{eff}$ are bounded, and thus $\mathrm{dom}\left((|V|_\mathrm{eff})^{1/2}\right)=\mathrm{dom}(|V_\mathrm{eff}|^{1/2})$ holds.
\end{proof}

\begin{proof}[Proof of Theorem \ref{thm;derivation}]
	This follows from Lemma \ref{lem;eff1}, Lemma \ref{lem;eff2} and Lemma \ref{lem;form_dom}.
\end{proof}

Roughly speaking, $H_\mathrm{eff}$ is the unique self-adjoint operator satisfying
\[
\langle\psi_\mathrm{dr},H_\mathrm{tot}\psi_\mathrm{dr}\rangle=\langle\psi,H_\mathrm{eff}\psi\rangle,\qquad\forall \psi\in L^2(\mathbb{R}_{\boldsymbol{x}}^d).
\]
However, both sides do not necessarily make sense.
To make this characterization rigorous, we have used the theory of quadratic forms.

\begin{corollary}
	It holds that
	\[
	\inf\sigma\left(-\frac{\hbar^2}{2m}\Delta\dot{+}V\right)\leq\inf\sigma(H_\mathrm{tot})\leq\inf\sigma(H_\mathrm{eff}),
	\]
where $\sigma(A)$ denotes the spectrum of an operator $A$.
\end{corollary}

\begin{proof}
	The second inequality follows from Theorem \ref{thm;derivation} and the variational principle.
	The first inequality follows from the same argument as in \cite[(3.24)]{MR1075749}; namely, it follows from \eqref{eq;diag of H_0}, the non-negativity of $\mathrm{d}\Gamma_\mathrm{b}(\hbar \omega)$, and the fact that $\Delta\otimes1$ commutes with $U$.
	
\end{proof}


\section{Concluding remarks}\label{sect;concluding_remarks}

We briefly comment on a variation of $H_0$ in which the electron mass in the interaction term is renormalized:
\[
H_0:=-\frac{\hbar^2}{2m_0}\Delta\otimes1+1\otimes\mathrm{d}\Gamma_\mathrm{b}(\hbar \omega)-\frac{q}{m_0}\sum_{j=1}^dp_j\otimes A_j(\boldsymbol{0}).
\]
In this case, the relation between the bare mass $m_0$ and the observed mass $m$ is given by
\[
\frac{1}{m}=\frac{1}{m_0}-\frac{d-1}{d}\frac{q^2}{\epsilon_0m_0^2}\int_{\mathbb{R}^d}\frac{|\hat{\rho}(\boldsymbol{k})|^2}{\omega(\boldsymbol{k})^2}\,\mathrm{d}\boldsymbol{k},
\]
where the right-hand side is assumed to be positive.
By the same arguments as above, the resulting effective Hamiltonian is given by \eqref{eq;def_of_eff_potential}--\eqref{eq;def_of_Heff} with the observed mass $m$ in \eqref{eq;def_of_a} replaced by the bare mass $m_0$.

From a physical point of view, one may expect that the mass in the interaction term should also be renormalized.
However, in this case, the bare mass still remains in the effective potential.
This may be due to the omission of the $\boldsymbol{A}^2$ term in the interaction.
In this regard, we note that scaling limits without neglecting the $\boldsymbol{A}^2$ term have been investigated by Hiroshima \cite{MR1235953,MR1438035,MR1892752,hiroshima2025weak}.

\section*{Acknowledgements}
I would like to thank Asao Arai and Itaru Sasaki for helpful discussions and for their careful reading of an earlier version of this manuscript.
This work was supported by JSPS KAKENHI Grant Numbers JP23K25783 and JP24K06755.

\section*{Conflicts of interest}
The author declares that they have no conflict of interest.

\section*{Data availability}
No data were used to support this study.

\section*{Use of AI}
During the preparation of this work, the author used ChatGPT (OpenAI) to improve the readability and language of the manuscript. 
The author reviewed and edited the content as needed and takes full responsibility for the content of the published article.

\bibliographystyle{plain}
\bibliography{References}

@article {MR1075749,
	AUTHOR = {A. Arai},
	TITLE = {An asymptotic analysis and its application to the
	nonrelativistic limit of the {P}auli-{F}ierz and a spin-boson
	model},
	JOURNAL = {J. Math. Phys.},
	FJOURNAL = {Journal of Mathematical Physics},
	VOLUME = {31},
	YEAR = {1990},
	NUMBER = {11},
	PAGES = {2653--2663},
	ISSN = {0022-2488,1089-7658},
	MRCLASS = {81Q10 (47N50 81V10 81V70)},
	MRNUMBER = {1075749},
	MRREVIEWER = {Giuseppina\ Epifanio},
	DOI = {10.1063/1.528966},
	URL = {https://doi.org/10.1063/1.528966},
}

@article {MR2770092,
	AUTHOR = {A. Arai},
	TITLE = {Spectral analysis of an effective {H}amiltonian in
	nonrelativistic quantum electrodynamics},
	JOURNAL = {Ann. Henri Poincar\'e},
	FJOURNAL = {Annales Henri Poincar\'e. A Journal of Theoretical and
	Mathematical Physics},
	VOLUME = {12},
	YEAR = {2011},
	NUMBER = {1},
	PAGES = {119--152},
	ISSN = {1424-0637,1424-0661},
	MRCLASS = {81V10 (47N50)},
	MRNUMBER = {2770092},
	MRREVIEWER = {Fumio\ Hiroshima},
	DOI = {10.1007/s00023-010-0071-2},
	URL = {https://doi.org/10.1007/s00023-010-0071-2},
}

@book {MR4812858,
	AUTHOR = {A. Arai},
	TITLE = {Analysis on {F}ock spaces and mathematical theory of quantum
	fields---an introduction to mathematical analysis of quantum
	fields},
	EDITION = {Second},
	PUBLISHER = {World Scientific Publishing Co. Pte. Ltd., Hackensack, NJ},
	YEAR = {2025},
	PAGES = {xxxix+1074},
	ISBN = {978-981-12-8842-5; 978-981-12-8845-6; 978-981-12-8844-9},
	MRCLASS = {81-01 (81-02 81Q10 81Q12 81R05 81R10 81T20 83C47)},
	MRNUMBER = {4812858},
}

@article{PhysRev.72.339,
	title = {The Electromagnetic Shift of Energy Levels},
	author = {H. A. Bethe},
	journal = {Phys. Rev.},
	volume = {72},
	issue = {4},
	pages = {339--341},
	numpages = {0},
	year = {1947},
	month = {Aug},
	publisher = {American Physical Society},
	doi = {10.1103/PhysRev.72.339},
	url = {https://link.aps.org/doi/10.1103/PhysRev.72.339}
}

@article {MR2015428,
	AUTHOR = {J. Derezi\'nski},
	TITLE = {Van {H}ove {H}amiltonians---exactly solvable models of the
	infrared and ultraviolet problem},
	JOURNAL = {Ann. Henri Poincar\'e},
	FJOURNAL = {Annales Henri Poincar\'e. A Journal of Theoretical and
	Mathematical Physics},
	VOLUME = {4},
	YEAR = {2003},
	NUMBER = {4},
	PAGES = {713--738},
	ISSN = {1424-0637,1424-0661},
	MRCLASS = {81U05 (47N50 81Q10 81T10)},
	MRNUMBER = {2015428},
	MRREVIEWER = {Fumio\ Hiroshima},
	DOI = {10.1007/s00023-003-0145-5},
	URL = {https://doi.org/10.1007/s00023-003-0145-5},
}

@article {MR1235953,
	AUTHOR = {F. Hiroshima},
	TITLE = {Scaling limit of a model of quantum electrodynamics},
	JOURNAL = {J. Math. Phys.},
	FJOURNAL = {Journal of Mathematical Physics},
	VOLUME = {34},
	YEAR = {1993},
	NUMBER = {10},
	PAGES = {4478--4518},
	ISSN = {0022-2488,1089-7658},
	MRCLASS = {81V10 (81Q10 81T10)},
	MRNUMBER = {1235953},
	MRREVIEWER = {Lech\ Jak\'obczyk},
	DOI = {10.1063/1.530353},
	URL = {https://doi.org/10.1063/1.530353},
}

@article {MR1438035,
	AUTHOR = {F. Hiroshima},
	TITLE = {A scaling limit of a {H}amiltonian of many nonrelativistic
	particles interacting with a quantized radiation field},
	JOURNAL = {Rev. Math. Phys.},
	FJOURNAL = {Reviews in Mathematical Physics. A Journal for Both Review and
	Original Research Papers in the Field of Mathematical Physics},
	VOLUME = {9},
	YEAR = {1997},
	NUMBER = {2},
	PAGES = {201--225},
	ISSN = {0129-055X,1793-6659},
	MRCLASS = {81V10 (81Q10)},
	MRNUMBER = {1438035},
	MRREVIEWER = {Hiroshi\ Isozaki},
	DOI = {10.1142/S0129055X97000075},
	URL = {https://doi.org/10.1142/S0129055X97000075},
}

@article {MR1892752,
	AUTHOR = {F. Hiroshima},
	TITLE = {Observable effects and parametrized scaling limits of a model
	in nonrelativistic quantum electrodynamics},
	JOURNAL = {J. Math. Phys.},
	FJOURNAL = {Journal of Mathematical Physics},
	VOLUME = {43},
	YEAR = {2002},
	NUMBER = {4},
	PAGES = {1755--1795},
	ISSN = {0022-2488,1089-7658},
	MRCLASS = {81V10 (81Q10)},
	MRNUMBER = {1892752},
	MRREVIEWER = {Masao\ Hirokawa},
	DOI = {10.1063/1.1447590},
	URL = {https://doi.org/10.1063/1.1447590},
}

@article{hiroshima2025weak,
	title={The weak coupling limit of the {P}auli-{F}ierz model},
	author={F. Hiroshima},
	journal={arXiv preprint arXiv:2505.07253},
	year={2025}
}

@article{PauliFierz1938,
	author = {W. Pauli and M. Fierz},
	title = {Zur Theorie der Emission langwelliger Lichtquanten},
	journal = {Nuovo Cimento},
	year = {1938},
	volume = {15},
	pages = {167--188},
	doi = {10.1007/BF02958939}
}

@book {MR493420,
	AUTHOR = {M. Reed and B. Simon},
	TITLE = {Methods of modern mathematical physics. {II}. {F}ourier
	analysis, self-adjointness},
	PUBLISHER = {Academic Press [Harcourt Brace Jovanovich, Publishers], New
	York-London},
	YEAR = {1975},
	PAGES = {xv+361},
	MRCLASS = {47-02 (81.47)},
	MRNUMBER = {493420},
	MRREVIEWER = {P.\ R.\ Chernoff},
}

@book {MR2953553,
	AUTHOR = {K. Schm\"{u}dgen},
	TITLE = {Unbounded self-adjoint operators on {H}ilbert space},
	SERIES = {Graduate Texts in Mathematics},
	VOLUME = {265},
	PUBLISHER = {Springer, Dordrecht},
	YEAR = {2012},
	PAGES = {xx+432},
	ISBN = {978-94-007-4752-4},
	MRCLASS = {47-01 (47B25 47E05)},
	MRNUMBER = {2953553},
	MRREVIEWER = {G. V. Rozenblum},
	DOI = {10.1007/978-94-007-4753-1},
	URL = {https://doi.org/10.1007/978-94-007-4753-1},
}

@book {MR455975,
	AUTHOR = {B. Simon},
	TITLE = {Quantum mechanics for {H}amiltonians defined as quadratic
	forms},
	SERIES = {Princeton Series in Physics},
	PUBLISHER = {Princeton University Press, Princeton, NJ},
	YEAR = {1971},
	PAGES = {xv+244},
	MRCLASS = {81.00},
	MRNUMBER = {455975},
	MRREVIEWER = {Thomas\ Spencer},
}

@article{PhysRev.74.1157,
	title = {Some Observable Effects of the Quantum-Mechanical Fluctuations of the Electromagnetic Field},
	author = {T. A. Welton},
	journal = {Phys. Rev.},
	volume = {74},
	issue = {9},
	pages = {1157--1167},
	numpages = {0},
	year = {1948},
	month = {Nov},
	publisher = {American Physical Society},
	doi = {10.1103/PhysRev.74.1157},
	url = {https://link.aps.org/doi/10.1103/PhysRev.74.1157}
}

\end{document}